\documentclass[10pt]{article}
\usepackage{amsmath, amssymb, amsthm, mathrsfs}
\usepackage{enumerate}
\usepackage{colordvi}
\usepackage{graphicx}

\newtheorem{claim}{Claim}[section]
\newtheorem{theorem}[claim]{Theorem}
\newtheorem{remark}[claim]{Remark}

\begin{document}
\begin{center}
{\Large{\textbf{Spectral estimates for Dirichlet Laplacians \\ [.3em]  on perturbed twisted tubes}}}

\bigskip

{\large{Pavel Exner and Diana Barseghyan}\footnote{The research was partially supported by the Czech Science Foundation within the project P203/11/0701 and RVO61389005. }}

\end{center}

\bigskip

\textbf{Abstract.} We investigate Dirichlet Laplacian in a straight twisted tube of a non-circular cross section, in particular, its discrete spectrum coming from a local slowdown of the twist. We prove a Lieb-Thirring-type estimate for the spectral moments and present two examples illustrating how the bound depends on the tube cross section.

\bigskip


\bigskip


\section{Introduction} \label{s: intro}
\setcounter{equation}{0}

Relations between geometry of a region and spectral properties of operators describing the corresponding dynamics are a trademark topic of mathematical physics. A lot of work was done in the last quarter a century about Dirichlet Laplacians in infinitely extended tubular regions. A particularly interesting class of problems concerns properties of twisted tubes.

The fact that twisting of a non-circular tube gives rise to an effective repulsive interaction was noted for the first time in \cite{CB96}, a proper mathematical meaning to this fact was given through an appropriate Hardy-type inequality \cite{EKK08}. On the other hand, while a periodic twist of an infinite straight tube raises the threshold of the essential spectrum, it is natural to expect that its \emph{local slowdown} will act as an attraction and could thus give rise to a discrete spectrum of the corresponding Dirichlet Laplacian; in \cite{EK05} it was demonstrated that it is indeed the case. The effect has been further investigated, in particular, the paper \cite{BKRS09} analyzed the asymptotic distribution of eigenvalues in case when the `slowdown' perturbation is infinitely extended and has prescribed decay properties.

In this paper we address another aspect of this spectral problem. We suppose for simplicity that the perturbation is compactly supported so that the discrete spectrum is finite, and we ask about bounds to eigenvalue moments. Using the dimension-reduction technique of Laptev and Weidl \cite{LW00} we derive Lieb-Thirring-type inequalities for moments of order $\sigma>1/2$; for $\sigma\ge 3/2$ these inequalities contain the optimal semiclassical constant. The role of the potential is at that played by the negative spectrum of an auxiliary operator acting on the cross section of the tube the coefficients of which carry the geometric information.

We accompany the general result by a pair of examples illustrating the dependence of the obtained bound on the tube cross section. Not surprisingly the estimate becomes trivial as the cross section turns to a circle; using an elliptical disc we show that the bound behaves like $\mathcal{O}\left(\varepsilon^{\sigma+1/2}\right)$ for small values of the eccentricity $\varepsilon$. The second example is maybe less self-evident: we show that a sufficiently `wild' cross section can make the bound arbitrarily large even if both the cross section area and its diameter remain bounded.

\section{The Lieb-Thirring type estimate} \label{s: LTest}
\setcounter{equation}{0}

In this section we first formulate the problem and state the main result, then we prove it and comment on the spectrum on an auxiliary operator which appears in the derived inequality.

\subsection{The main result} \label{ss: main}

Let $\omega$ be an open set, bounded and connected in $\mathbb{R}^2$, and let $\theta:\, \mathbb{R} \to \mathbb{R}$ be a
differentiable function. For a given $s\in\mathbb{R}$ and $t:=(t_2,t_3)\in\omega$ we define the mapping $\mathfrak{L}:\: \mathbb{R}\times\omega\to\mathbb{R}^3$ by
\begin{equation}
\label{eq1}
\mathfrak{L}(s,t)=(s,t_2\cos\theta(s)+t_3\sin\theta(s),t_3\cos\theta(s)-t_2\sin\theta(s))\,,
\end{equation}
which has an obvious meaning of rotating the coordinate frame in the normal plane to the $x$-axis at the point $s$ by the angle $\theta(s)$. The image $\mathfrak{L}(\mathbb{R}\times\omega)$ is a region in $\mathbb{R}^3$ which we call a \emph{twisted tube}
if the following two conditions are satisfied:

\begin{enumerate}[(i)]
 \setlength{\itemsep}{-3pt}

\item the function $\theta$ is not constant,

\item the set $\omega$ is not rotationally symmetric with respect to the origin in $\mathbb{R}^2$;

\end{enumerate}

\noindent in particular, the tube is set to be \emph{periodically twisted} if $\theta$ is a linear function.

As we have indicated we are interested in tubes with a local perturbation of the periodic twisting. Consequently, we shall consider angular function $\theta$ such that
\begin{equation}
\label{eq2}
\dot{\theta}(s)=\beta(s)= \beta_0-\mu(s)\,,
\end{equation}
where $\beta_0$ is a positive constant\footnote{The positivity assumption is made just for the sake of simplicity since for $\beta_0<0$ one can pass to a unitarily equivalent operator using mirror-image transformation $(s,t)\mapsto (-s,t)$.} and $\mu(\cdot)$ is a positive and bounded function with $\mathrm{supp}\,\mu\subset[-s_0,s_0]$ for some $s_0>0$. We will concerned with tubes
\begin{equation}
\label{Tube}
\Omega_\beta:=\mathfrak{L}(\mathbb{R}\times\omega)
\end{equation}
defined by $\mathfrak{L}= \mathfrak{L}_\theta$ corresponding to $\theta(s)= \int_{-s_0}^s \beta(s)\, \mathrm{d}s$, the angle being determined up to a constant. The object of our study is the Dirichlet Laplacian $-\Delta_D^{\Omega_\beta}$ on $L^2\big(\Omega_\beta\big)$ associated with the quadratic form
$$
Q_\beta[\psi]=\int_{\Omega_\beta}|\nabla\psi(x)|^2\, \mathrm{d}x\,, \qquad \psi\in D(Q_\beta)=\mathcal{H}_0^1
\big(\Omega_\beta\big)\,.
$$
It is well know \cite{EK05, EKK08} that analysis of $-\Delta_D^{\Omega_\beta}$ can be rephrased as investigation of the unitarily equivalent operator $H_\beta$ on the `straightened tube', i.e. the cylindrical region $\mathbb{R}\times\omega$, which is associated with form
\begin{equation}
\label{straightform}
Q^0_\beta[\psi]=\int_{\mathbb{R}\times\omega} \left[ |(\nabla'\psi)(s,t)|^2 +\big|(\partial_s\psi+\dot{\theta}\partial_\tau\psi\big)(s,t)|^2 \right]\,\mathrm{d}s\,\mathrm{d}t\,,
\end{equation}
where
$$
\nabla':=\left(\frac{\partial}{\partial t_2},\,\frac{\partial}{\partial t_3}\right)\,,\qquad \partial_s:= \frac{\partial}{\partial s}\,, \quad\partial_\tau:=t_2\frac{\partial}{\partial t_3}-t_3\frac{\partial}{\partial t_2}\,,
$$
the latter being nothing else, up to an imaginary unit, than the angular momentum operator in the normal plane to tube axis.
In other words, the operator $H_\beta$ acts on its domain in $L^2\left(\mathbb{R}\times\omega\right)$ as
$$
H_\beta=-\frac{\partial^2}{\partial t_2^2}-\frac{\partial^2}{\partial t_3^2}+\left(-i\,\frac{\partial}
{\partial s}-i\,\dot{\theta}(s)\left(t_2\frac{\partial}{\partial t_3}-t_3\frac{\partial}{\partial t_2}
\right)\right)^2
$$
and the geometry of the original problem is translated into the coefficients of the `straightened' one. We are going to compare $H_\beta$ with the `unperturbed' operator $H_{\beta_0}$. The formulate the result, we need the following operator on $L^2(\omega)$,
$$
h_{\beta_0}=-\Delta_D^\omega-\beta_0^2\left(t_2\frac{\partial}{\partial t_3}-t_3\frac{\partial}
{\partial t_2}\right)^2 \quad \mbox{with} \quad  D(h_{\beta_0}) = \mathcal{H}^2(\omega) \cap \mathcal{H}^1_0(\omega)\,.
$$
It is convenient to introduce the polar coordinates $(r,\alpha)$ in the normal plane in which the action of $h_{\beta_0}$ can be written as
\begin{equation}
\label{eq5}
h_{\beta_0}=-\Delta_D^\omega-\beta_0^2\partial_\alpha^2\,,
\end{equation}
where the standard polar-coordinate expression can used for the Dirichlet Laplacian $-\Delta_D^\omega$. As the latter has a purely discrete spectrum and the second term is positive, the spectrum of $h_{\beta_0}$ is purely discrete as well. We denote
\begin{equation}
\label{eq6}
E:=\inf\sigma(h_{\beta_0})\,,
\end{equation}
further we consider the radius of $\omega$,
$$
d:=\sup_{(t_2,t_3)\in\omega}\sqrt{t_2^2+t_3^2}
$$
and the quantity
\begin{equation}
\label{gamma}
\gamma_{\beta_0}:=\min\left\{\frac{1}{3},\frac{1}{48\beta_0^2 d^2}\right\}.
\end{equation}

\medskip

\noindent Now we are in position to state our main result:

\begin{theorem} \label{th: main}
Let $H_\beta$ be the operator associated with form (\ref{straightform}) with the twisting velocity $\dot{\theta}$ given by  (\ref{eq2}), and assume that the compactly supported function $\mu(\cdot)$ is bounded with the derivative $\dot\mu \in L^1$ and satisfies the condition $\|\mu\|_\infty<c\,\beta_0$ with some $c\in(0,\frac13 \gamma_{\beta_0})$. Then the following Lieb-Thirring-type inequality,
\begin{equation}
\label{Inequality}
\mathrm{tr}\left(H_\beta-E\right)_-^\sigma\le\alpha_{\mu,\beta_0}^{2\sigma}\,L_{\sigma,1}^{\emph{cl}}
\int_{\mathbb{R}}\mathrm{tr}\,H(s)_-^{\sigma+1/2}\,\mathrm{d}s\,,\quad\sigma\ge3/2\,,
\end{equation}
is valid for the discrete spectrum of $H_\beta$, where
\begin{equation}
\label{alpha}
\alpha_{\mu,\beta_0}^2:=\gamma_{\beta_0}-3c\,,
\end{equation}
and furthermore, $H(s)_-$ is the negative part of the operator
\begin{equation}
\label{operator}
-\Delta_D^\omega-\beta_0^2\left(t_2\frac{\partial}{\partial t_3}-t_3\frac{\partial}{\partial t_2}
\right)^2-\frac{1}{\alpha_{\mu,\beta_0}^2}\left(\dot{\mu}\, \frac{\partial_\alpha f}{f}\,-\mu(2\beta_0-\mu) \frac{\partial_\alpha^2f}{f}\right)-E
\end{equation}
with the domain $\mathcal{H}^2(\omega) \cap \mathcal{H}^1_0(\omega)$, where $-\Delta_D^\omega$ is the Dirichlet Laplacian, $f$ is the ground-state eigenfunction of $h_{\beta_0}$, and
 \begin{equation} \label{LTconstant}
L_{\sigma,1}^{\mathrm{cl}} := \frac{\Gamma(\sigma+1)}{\sqrt{4\pi}\, \Gamma(\sigma+\frac32)}\,,
 \end{equation}
is the standard semiclassical constant.
\end{theorem}
\begin{remark} \label{moresigma}
{\rm The validity of the result can be in analogy with \cite{ELW04} extended to any $\sigma\ge 1/2$; the price to pay is only a change in the constant, $L_{\sigma,1}^{\mathrm{cl}}$ being replaced with $r(\sigma,1) L_{\sigma,1}^{\mathrm{cl}}$ with the factor $r(\sigma,1)\le 2$ if $\sigma<3/2$.
}
\end{remark}
\begin{proof}
Using the fact that the ground-state eigenfunction $f$ of $h_{\beta_0}$ is strictly positive in $\omega$ --- cf.~\cite{EK05} --- we can decompose any $\psi\in C_0^\infty(\omega)$ as
$$
\psi(s,t)=f(t)\varphi(s,t)\,.
$$
We substitute this into (\ref{straightform}) and integrate twice by parts obtaining
\begin{eqnarray*}
\lefteqn{Q_{\beta}^0[\psi]-E\|\psi\|^2=\int_{\mathbb{R}\times\omega}\Big(f^2|\nabla'\varphi|^2
-(\Delta_D^\omega f)f|\varphi|^2+f^2|\partial_s\varphi|^2} \\[.5em] && +(\beta_0-\mu)f\partial_\alpha f(\partial_s\overline{\varphi}\,
\varphi+\overline{\varphi}\,\partial_s\varphi)
+(\beta_0-\mu)f^2(\partial_s\overline{\varphi}\,\partial_\alpha\varphi+\partial_s\varphi\,\partial_\alpha
\overline{\varphi}) \\[.5em] &&
+(\beta_0-\mu)^2f^2|\partial_\alpha\varphi|^2
-(\beta_0-\mu)^2(\partial_\alpha^2f)
f|\varphi|^2-E\,f^2|\varphi|^2\Big)\,\mathrm{d}s\,\mathrm{d}t\,.
\end{eqnarray*}
Since
$$
\int_{\mathbb{R}}(\partial_s\overline{\varphi}\,\varphi+\overline{\varphi}\,\partial_s\varphi)\, \mathrm{d}s=0\,, \quad
\int_{\mathbb{R}}\mu(\partial_s\overline{\varphi}\,\varphi+\overline{\varphi}\,\partial_s\varphi)\,
\mathrm{d}s=-\int_{\mathbb{R}}\dot{\mu}\,|\varphi|^2\,\mathrm{d}s\,,
$$
where the last integral makes sense in view of boundedness of the function  $\varphi$, and $-\Delta_D^\omega f-\beta_0^2\partial_\alpha^2f-E\,f=0$ holds by assumption, we get
\begin{eqnarray}
\lefteqn{Q_\beta^0[\psi]-E\|\psi\|^2=\int_{\mathbb{R}\times\omega}f^2\big(|\nabla'\varphi|^2+|\partial_s\varphi
+\beta_0\partial_\alpha\varphi|^2\big) \,\mathrm{d}s\,\mathrm{d}t} \nonumber
\\[.3em] &&
\;\;-\int_{\mathbb{R}\times\omega}\biggl(\mu\,f\partial_\alpha f(\partial_s\overline{\varphi}\,\varphi+\overline{\varphi}\,
\partial_s\varphi)+\mu\,f^2(\partial_s\overline{\varphi}\,\partial_\alpha\varphi+\partial_\alpha
\overline{\varphi}\,\partial_s\varphi) \nonumber
\\[.3em] &&
\;\;-(\mu^2-2\beta_0\mu)f^2|\partial_\alpha\varphi|^2
+(\mu^2-2\beta_0\mu)(\partial_\alpha^2f)f|\varphi|^2\biggr)\,\mathrm{d}s\,\mathrm{d}t \nonumber
\\[.3em] &&
\label{eq7}
=\int_{\mathbb{R}\times\omega}
f^2\left(|\nabla'\varphi|^2+|\partial_s\varphi+\beta_0\partial_\alpha\varphi|^2\right)\,\mathrm{d}s\,\mathrm{d}t-I_\beta\,,
\end{eqnarray}
where
\begin{eqnarray*}
\lefteqn{I_\beta:=\int_{\mathbb{R}\times\omega}\biggl(\dot{\mu}\,f\partial_\alpha f |\varphi|^2+\mu\,f^2(\partial_s
\overline{\varphi}\,\partial_\alpha\varphi+\partial_\alpha\overline{\varphi}\,\partial_s\varphi)} \\[.3em] &&
-\mu(\mu-2\beta_0)\,f^2|\partial_\alpha\varphi|^2
+\mu(\mu-2\beta_0)(\partial_\alpha^2f)f|\varphi|^2\biggr)\,\mathrm{d}s\,\mathrm{d}t\,;
\end{eqnarray*}
the integral again makes sense since $\dot\mu \in L^1$ by assumption and the other involved functions are bounded.
Let us estimate it. Using the Cauchy inequality one gets
\begin{eqnarray}
\lefteqn{\nonumber
I_\beta\le\int_{\mathbb{R}\times\omega}\dot{\mu}\,f(\partial_\alpha f)|\varphi|^2\,\mathrm{d}s\,\mathrm{d}t}
\\[.3em] &&
\nonumber
\;\;+2\left(\int_{\mathbb{R}\times\omega}\mu^2\,f^2|\partial_\alpha\varphi|^2\,\mathrm{d}s\,\mathrm{d}t\right)^{1/2}
\left(\int_{\mathbb{R}\times\omega}f^2|\partial_s\varphi|^2\,\mathrm{d}s\,\mathrm{d}t\right)^{1/2}
\\[.3em] &&
\nonumber
\;\;+2\beta_0\int_{\mathbb{R}\times\omega}\mu\,f^2|\partial_\alpha\varphi|^2\,\mathrm{d}s\,\mathrm{d}t
+\int_{\mathbb{R}\times\omega}(\mu^2-2\beta_0\mu)\,\frac{\partial_\alpha^2f}{f}|\varphi|^2\,\mathrm{d}s\,\mathrm{d}t
\\[.3em] &&
\nonumber
\le\int_{\mathbb{R}\times\omega}\dot{\mu}\,\frac{\partial_\alpha f}{f}|\psi|^2\,\mathrm{d}s\,\mathrm{d}t
+\int_{\mathbb{R}\times\omega}(\mu^2-2\beta_0\mu)\,\frac{\partial^2_\alpha f}{f}|\psi|^2\,\mathrm{d}s\,\mathrm{d}t
\\[.3em] &&
\nonumber
\;\;+c\int_{\mathbb{R}\times\omega}f^2|\partial_s\varphi|^2\,\mathrm{d}s\,\mathrm{d}t
+\int_{\mathbb{R}\times\omega} \big(c^{-1} + 2\beta_0\mu \big) \mu^2\,f^2
|\partial_\alpha\varphi|^2\,\mathrm{d}s\,\mathrm{d}t
\\[.3em] &&
\label{eq8}
\le\int_{\mathbb{R}\times\omega}\dot{\mu}\,\frac{\partial_\alpha f}{f}|\psi|^2\,\mathrm{d}s\,\mathrm{d}t
+\int_{\mathbb{R}\times\omega}(\mu^2-2\beta_0\mu)\,\frac{\partial^2_\alpha f}{f}|\psi|^2\,\mathrm{d}s\,\mathrm{d}t
\\[.3em] &&
\nonumber
\;\;+c\int_{\mathbb{R}\times\omega}f^2|\partial_s\varphi|^2\,\mathrm{d}s\,\mathrm{d}t
+\left( \frac{\|\mu\|_\infty^2}{c\beta_0^2} + \frac{2\|\mu\|_\infty}{\beta_0} \right)
\int_{\mathbb{R}\times\omega}\beta_0^2\,f^2|\partial_\alpha\varphi|^2\,\mathrm{d}s\,\mathrm{d}t
\end{eqnarray}
for an arbitrary $c>0$. Let us now return to the quadratic form expression (\ref{eq7}). If
$$
2\beta_0\int_{\mathbb{R}\times\omega}f^2|\partial_s\varphi||\partial_\alpha\varphi|\,\mathrm{d}s\,\mathrm{d}t\le\frac{1}{2}
\int_{\mathbb{R}\times\omega}f^2|\partial_s\varphi|^2\,\mathrm{d}s\,\mathrm{d}t
$$
then the said formula yields the estimate
\begin{equation}
\label{eq9}
Q_\beta^0[\psi]-E\|\psi\|^2\ge\frac{1}{2}\int_{\mathbb{R}\times\omega}f^2\left(|\nabla'\varphi|^2
+|\partial_s\varphi|^2+\beta_0^2|\partial_\alpha\varphi|^2\right)\,\mathrm{d}s\,\mathrm{d}t-I_\beta\,.
\end{equation}
In the opposite case, namely under the assumption
$$
2\beta_0\int_{\mathbb{R}\times\omega}f^2|\partial_s\varphi||\partial_\alpha\varphi|\,\mathrm{d}s\,\mathrm{d}t>\frac{1}{2}
\int_{\mathbb{R}\times\omega}f^2|\partial_s\varphi|^2\,\mathrm{d}s\,\mathrm{d}t\,,
$$
one has to employ the obvious inequality $|\partial_\alpha\varphi|\le d\,|\nabla'\varphi|$ to obtain
$$
\int_{\mathbb{R}\times\omega}f^2|\partial_s\varphi|^2\,\mathrm{d}s\,\mathrm{d}t<16\beta_0^2d^2\int_{\mathbb{R}\times\omega}f^2
|\nabla'\varphi|^2\,\mathrm{d}s\,\mathrm{d}t\,.
$$
Combining these estimates with the definition (\ref{gamma}), we infer from (\ref{eq7}) that
\begin{equation}
\label{LTPerturb}
Q_\beta^0[\psi]-E\|\psi\|^2\ge\gamma_{\beta_0}\int_{\mathbb{R}\times\omega}f^2\left(|\nabla'\varphi|^2
+|\partial_s\varphi|^2+\beta_0^2|\partial_\alpha\varphi|^2\right)\,\mathrm{d}s\,\mathrm{d}t-I_\beta\,.
\end{equation}
Furthermore, in view of inequality (\ref{eq8}) it follows from (\ref{LTPerturb}) that
\begin{eqnarray*}
\lefteqn{\nonumber
Q_\beta^0[\psi]-E\|\psi\|^2\ge\alpha_{\mu,\beta_0}^2\int_{\mathbb{R}\times\omega}f^2\left(|\nabla'\varphi|^2+
|\partial_s\varphi|^2+\beta_0^2|\partial_\alpha\varphi|^2\right)\,\mathrm{d}s\,\mathrm{d}t}
\\[.3em] &&
-\int_{\mathbb{R}\times\omega}\left(\dot{\mu}\,\frac{\partial_\alpha f}{f}-\mu\,(2\beta_0-\mu)\, \frac{\partial^2_\alpha f}{f}\right)|\psi|^2\,\mathrm{d}s\,\mathrm{d}t \phantom{AAAAAA}
\end{eqnarray*}
with the number $\alpha_{\mu,\beta_0}$ given in (\ref{alpha}). In the next step we estimate the integrals $\int_{\mathbb{R}\times\omega}f^2|\nabla'\varphi|^2\,\mathrm{d}s\,\mathrm{d}t$ and $\int_{\mathbb{R}\times\omega}f^2|\partial_\alpha\varphi|^2\,\mathrm{d}s\, \mathrm{d}t$. An integration by parts yields
\begin{eqnarray*}
\lefteqn{\nonumber
\int_{\mathbb{R}\times\omega}f^2\left|\frac{\partial\varphi}{\partial t_j}\right|^2\,\mathrm{d}s\,
\mathrm{d}t=\int_{\mathbb{R}\times\omega}\left|\frac{\partial\psi}{\partial t_j}-\varphi\frac{\partial f}
{\partial t_j}\right|^2\,\mathrm{d}s\,\mathrm{d}t=
\int_{\mathbb{R}\times\omega}\left|\frac{\partial\psi}{\partial t_j}\right|^2\,\mathrm{d}s\,\mathrm{d}t}
\\[.3em] && 
+\int_{\mathbb{R}\times\omega}\frac{1}{f^2}\left(\frac{\partial f}{\partial t_j}\right)^2|\psi|^2\,
\mathrm{d}s\,\mathrm{d}t+\int_{\mathbb{R}\times\omega}\frac{\partial}{\partial t_j}\left(\frac{1}{f}
\frac{\partial f}{\partial t_j}\right)|\psi|^2\,\mathrm{d}s\,\mathrm{d}t \phantom{AAAA}
\end{eqnarray*}
for $j=2,3$. In a similar way, changing integration variables one finds
\begin{eqnarray*}
\lefteqn{\nonumber
\int_{\mathbb{R}\times\omega}f^2|\partial_\alpha\varphi|^2\,\mathrm{d}s\,\mathrm{d}t=\int_{\mathbb{R}\times\omega}
\left|\partial_\alpha\psi-\varphi\partial_\alpha f\right|^2\,\mathrm{d}s\,\mathrm{d}t
=\int_{\mathbb{R}\times\omega}|\partial_\alpha\psi|^2\,\mathrm{d}s\,\mathrm{d}t}
\\[.3em] &&
+\int_{\mathbb{R}\times\omega}\frac{(\partial_\alpha f)^2}{f^2}|\psi|^2\,\mathrm{d}s\,\mathrm{d}t+\int_{\mathbb{R}\times\omega}\partial_\alpha\left(\frac{\partial_\alpha f}{f} \right) |\psi|^2\,\mathrm{d}s\,\mathrm{d}t\,, \phantom{AAA}
\end{eqnarray*}
hence inequality (\ref{LTPerturb}) implies
\begin{eqnarray*}
\lefteqn{\nonumber
Q_\beta^0[\psi]-E\|\psi\|^2\ge\alpha_{\mu,\beta_0}^2\int_{\mathbb{R}\times\omega}\left(|\nabla'\psi|^2+
|\partial_s\psi|^2+\beta_0^2|\partial_\alpha\psi|^2\right)\,\mathrm{d}s\,\mathrm{d}t}
\\[.3em] &&
\nonumber
+\alpha_{\mu,\beta_0}^2\int_{\mathbb{R}\times\omega}\biggl(\left|\nabla'f\right|^2/f^2+\frac{\partial}{\partial t_2}
\left(\frac{1}{f}\frac{\partial f}{\partial t_2}\right)+\frac{\partial}{\partial t_3}
\left(\frac{1}{f}\frac{\partial f}{\partial t_3}\right)
\\[.3em] &&
\nonumber
\quad+\beta_0^2\frac{(\partial_\alpha f)^2}{f^2}+\beta_0^2\partial_\alpha\left(\frac{\partial_\alpha f}{f}\right)\biggr)|\psi|^2\,\mathrm{d}s\,\mathrm{d}t
\\[.3em] &&
-\int_{\mathbb{R}\times\omega}\left(\dot{\mu}\,\frac{\partial_\alpha f}{f}-\mu(2\beta_0-\mu)\,\frac{\partial^2_\alpha f}{f}\right)|\psi|^2\,\mathrm{d}s\,\mathrm{d}t\,.
\end{eqnarray*}
After a differentiation we get from here
\begin{eqnarray*}
\lefteqn{Q_\beta^0[\psi]-E\|\psi\|^2} \\[.3em] &&
\ge\alpha_{\mu,\beta_0}^2\int_{\mathbb{R}\times\omega}\left(|\nabla'\psi|^2+
|\partial_s\psi|^2+\beta_0^2|\partial_\alpha\psi|^2+
\frac{\Delta_D^\omega f}{f}+\beta_0^2\frac{\partial_\alpha^2f}{f}\right)\,\mathrm{d}s\,\mathrm{d}t \\[.3em] &&
-\int_{\mathbb{R}\times\omega}\left(\dot{\mu}\,\frac{\partial_\alpha f}{f}-\mu(2\beta_0-\mu)\,\frac{\partial_\alpha^2f}{f}\right)|\psi|^2\,\mathrm{d}s\,\mathrm{d}t\,,
\end{eqnarray*}
so in view of the fact that $f$ is the ground-state eigenfunction of operator $h_{\beta_0}$ we are able to conclude that
\begin{eqnarray}
\lefteqn{\nonumber
Q_\beta^0[\psi]-E\|\psi\|^2\ge\alpha_{\mu,\beta_0}^2\int_{\mathbb{R}\times\omega}\left(|\nabla'\psi|^2+
|\partial_s\psi|^2+\beta_0^2|\partial_\alpha\psi|^2-E|\psi|^2\right)\,\mathrm{d}s\,\mathrm{d}t}
\\[.3em] &&
\label{perturbation}
-\int_{\mathbb{R}\times\omega}\left(\dot{\mu}\,\frac{\partial_\alpha f}{f}-\mu(2\beta_0-\mu)\,\frac{\partial_\alpha^2f}{f}\right)|\psi|^2\,\mathrm{d}s\,\mathrm{d}t\,. \phantom{AAAAAAAAAAAA}
\end{eqnarray}
Up to the numerical factor the quadratic form on the left-hand side corresponds to the self-adjoint operator $H_0^-$ on $L^2(\mathbb{R}\times\omega)$ acting as
\begin{eqnarray*}
H_0^-:=-\Delta_D^{\mathbb{R}\times\omega}-\beta_0^2\left(t_2\frac{\partial}
{\partial t_3}-t_3\frac{\partial}{\partial t_2}\right)^2
-\frac{1}{\alpha_{\mu,\beta_0}^2}\left(\dot{\mu}\,\frac{\partial_\alpha f}{f}-\mu(2\beta_0-\mu)\,\frac{\partial_\alpha^2f}{f}\right)-E\,;
\end{eqnarray*}
we are going to prove that its negative spectrum is discrete and establish a bound to it. We shall employ (\ref{perturbation}) in combination with the minimax principle.

Choosing for $\psi$ a function $u\in\,C_0^\infty(\mathbb{R}\times\omega)$ we can estimate the right-hand side of (\ref{perturbation}) from below as follows
\begin{eqnarray*}
\lefteqn{\int_{\mathbb{R}\times\omega}\left(|\partial_s u|^2+|\nabla'u|^2+\beta_0^2|\partial_\alpha u|^2-E|u|^2\right)\,
\mathrm{d}s\,\mathrm{d}t}
\\[.3em] &&
\quad -\frac{1}{\alpha_{\mu,\beta_0}^2}\int_{\Omega_0}\left(\dot{\mu}\,\frac{\partial_\alpha f}{f}-\mu(2\beta_0-\mu)\,\frac{\partial_\alpha^2f}{f}\right)|u|^2\,\mathrm{d}s\,\mathrm{d}t
\\[.3em] &&
\ge\int_{\mathbb{R}\times\omega}|\partial_su|^2\,\mathrm{d}s\,\mathrm{d}t
+\int_{\mathbb{R}}\left\langle\,H(s)\,u(s,\cdot),\,u(s,\cdot)\right\rangle_{L^2(\omega)}\,
\mathrm{d}s\,,
\end{eqnarray*}
where $H(s)$ is the negative part of Schr\"{o}dinger operator (\ref{operator}).Consider next functions of the form $g=u+v$, where $u$ as above and $v\in\,C_0^\infty(\widehat{\mathbb{R}\times\omega})$ supported by the complement $\widehat{\mathbb{R} \times\omega}:= \mathbb{R}^3\backslash\overline{\mathbb{R}\times\omega}$ are both extended by zero to $\mathbb{R}^3$. We have
\begin{eqnarray}
\lefteqn{\nonumber
\|\partial_su\|^2_{L^2(\mathbb{R}\times\omega)}+\int_{\mathbb{R}\times\omega}\left(|\nabla'u|^2+\beta_0^2|\partial_\alpha u|^2-E|u|^2\right)\,\mathrm{d}s\,\mathrm{d}t
+\|\partial_s v\|_{L^2(\widehat{\mathbb{R}\times\omega})}^2}
\\[.3em] &&
\nonumber
\quad +\|\nabla'v\|^2_{L^2(\widehat{\mathbb{R}\times\omega})}
-\frac{1}{\alpha_{\mu,\beta_0}^2}\int_{\mathbb{R}\times\omega}\left(\dot{\mu}\,\frac{\partial_\alpha f}{f}-\mu(2\beta_0-\mu)\frac{\partial_\alpha^2f}{f}\right)|u|^2\,\mathrm{d}s\,\mathrm{d}t
\\[.3em] &&
\label{SP}
\ge\int_{\mathbb{R}^3}|\partial_sg|^2\,\mathrm{d}s\,\mathrm{d}t+\int_{\mathbb{R}}\left\langle\,H(s)
\,g(s,\cdot),\,g(s,\cdot)\right\rangle_{L^2(\mathbb{R}^2)}\,\mathrm{d}s\,,
\end{eqnarray}
where we have extended $H(s)$ to an orthogonal sum acting as zero on $C_0^\infty\left(\mathbb{R}^2\backslash\, \overline{\omega}\right)$ obtaining thus an operator on $C_0^\infty(\mathbb{R}^2)$. The inequality (\ref{SP}) holds true for any
function $g\in\,C_0^\infty\left(\mathbb{R}^3\backslash\partial\Omega_0\right)$ and its left-hand side is the quadratic form associated with the operator $H_0^-\oplus\left(-\Delta_D^{\widehat{\mathbb{R} \times\omega}}\right)$. On the other hand, the right-hand side of (\ref{SP}) is the form associated with the operator $-\frac{\partial^2}{\partial\,s^2}\otimes\, I_{L^2(\mathbb{R}^2)}+H(s)$ defined on the larger domain $\mathcal{H}^2\left(\mathbb{R},L^2(\mathbb{R}^2)\right)$.
Due to the positivity of the Dirichlet Laplacian $-\Delta_D^{\widehat{\mathbb{R}\times\omega}}$ the variational principle allows us to conclude
\begin{equation}
\label{eq4.9}
\mathrm{tr}\,\left(H^-_0\right)_-^\sigma\le\,\mathrm{tr}\left(-\frac{\partial^2}{\partial\,s^2}
\otimes\,I_{L^2(\mathbb{R}^2)}+H(s)\right)_-^\sigma,\quad\sigma\ge0\,.
\end{equation}
Then in turn the Lieb-Thirring inequality for operator-valued potentials \cite{LW00} implies
$$
\mathrm{tr}\,\left(H^-_0\right)_-^\sigma\le\,L_{\sigma,1}^{\emph{cl}}
\int_{\mathbb{R}}\mathrm{tr}\,H(s)_-^{\sigma+1/2}\,\mathrm{d}s \quad\text{for any}\quad\sigma\ge3/2
$$
with the semiclassical constant $L_{\sigma,1}^{\emph{cl}}$ (and for $\sigma\ge 1/2$ with a worse constant as mentioned in Remark~\ref{moresigma}). By minimax principle we thus finally get
$$
\mathrm{tr}\left(H_\beta-E\right)_-^\sigma\le\alpha_{\mu,\beta_0}^{2\sigma}\,\,L_{\sigma,1}^{\emph{cl}}
\int_{\mathbb{R}}\mathrm{tr}\,H(s)_-^{\sigma+1/2}\,\mathrm{d}s \quad\text{for}\quad\sigma\ge3/2\,,
$$
which concludes the argument.
\end{proof}
\bigskip

\subsection{The discreteness of the spectrum of operator $H(s)$} \label{ss: discr}

In order to apply Theorem~\ref{th: main} we have to make sure that the (negative part of the) spectrum of the operator $H(s)$ is discrete. It looks almost self-evident, but we present a proof, under slightly stronger assumptions:

\begin{theorem}
The spectrum of $H(s)$ is purely discrete for any fixed $s\in\mathbb{R}$ if the following condition is satisfied:
\begin{equation}
\label{sp}
\max\left\{\frac{2\|\mu\|_\infty}{\beta_0},\,
\frac{\|\dot{\mu}\|_\infty}{2\beta_0^2}\right\}<\alpha_{\mu,\beta_0}^2.
\end{equation}
\end{theorem}
\begin{proof}
We shall estimate $H(s)$ from below by an operator the spectrum of which is purely discrete, then the claim will follow by the minimax principle. Using once more the fact that the ground-state eigenfunction is strictly positive, we can represent any $\varphi\in \mathcal{H}^2(\omega) \cap \mathcal{H}^1_0(\omega)$ as $\varphi(t)=f(t)\,\phi(t),\:t\in\omega$. Using an integration by parts in the angular variable in combination with Cauchy inequality we get
\begin{eqnarray}
\lefteqn{\nonumber
\left|\int_\omega \frac{\partial_\alpha f}{f}\, |\varphi|^2\,\mathrm{d}t\right|=
\left|\int_\omega f(\partial_\alpha f)|\phi|^2\,\mathrm{d}t\right|=\frac{1}{2}\left|\int_\omega f^2(\partial_\alpha\overline{\phi}\,\phi+\partial_\alpha\phi\,\overline{\phi})
\,\mathrm{d}t\right|}
\\[.3em] &&
\nonumber
\le\left(\int_\omega f^2|\partial_\alpha\phi|^2\,\mathrm{d}t\right)^{1/2}\left(\int_\omega f^2|\phi|^2\,\mathrm{d}t\right)^{1/2}
\\[.3em] &&
\label{perturb1}
\le\frac{1}{2}\int_\omega f^2|\partial_\alpha\phi|^2\mathrm{d}t+\frac{1}{2}\int_\omega|\varphi|^2\,\mathrm{d}s\,\mathrm{d}t\,. \phantom{AAAAAAAAAAAAAA}
\end{eqnarray}
Next we use the fact that the function $f$ is ground state-eigenfunction of operator $h_\beta$. Combining another integration by parts with Cauchy inequality and denoting here for the sake of brevity $(x,y)=(x_1,x_2)$ we derive
\begin{eqnarray}
\lefteqn{\nonumber
\beta_0^2\int_\omega\frac{\partial^2_\alpha f}{f}|\varphi|^2\,\mathrm{d}t=-\beta_0^2\int_\omega(\partial_\alpha^2 f)f|\phi|^2\,\mathrm{d}t
=\int_\omega(Ef^2|\phi|^2\,\mathrm{d}t+(\Delta_D^\omega f)f|\phi|^2)\,\mathrm{d}t}
\\[.3em] &&
\nonumber
=E\int_\omega f^2|\phi|^2\,\mathrm{d}t+\int_\omega(\Delta_D^\omega)f|\phi|^2\,\mathrm{d}t=E\int_\omega |\varphi|^2\,\mathrm{d}t+\int_\omega(\Delta_D^\omega)f|\psi|^2\,\mathrm{d}t
\\[.3em] &&
\nonumber
=E\int_\omega|\varphi|^2\,\mathrm{d}t
+\sum_{j=1}^2 \int_\omega\frac{\partial^2f}{\partial x^2_j}\,f|\phi|^2\,\mathrm{d}t
=E\int_\omega|\varphi|^2\,\mathrm{d}t+ \sum_{j=1}^2\int_\omega f|\phi|^2\frac{\mathrm{d}}{\mathrm{d}x_j}\left(\frac{\partial f}{\partial x_j}\right)\,\mathrm{d}y
\\[.3em] &&
\nonumber
=E\int_\omega|\varphi|^2\,\mathrm{d}t- \sum_{j=1}^2 \int_\omega\left(\frac{\partial f}{\partial x_j}\right)^2|\phi|^2\,\mathrm{d}t- \sum_{j=1}^2 \int_\omega\frac{\partial f}{\partial x_j}f\left(\frac{\partial\phi}{\partial x_j}\overline{\phi}+\frac{\partial\overline{\phi}}{\partial x_j}\phi\right)\,\mathrm{d}t
\\[.3em] &&
\nonumber
\le E\int_\omega|\varphi|^2\,\mathrm{d}t-\int_\omega|\nabla'f|^2|\phi|^2\,\mathrm{d}t+2 \sum_{j=1}^2\int_\omega f|\phi|\left|\frac{\partial f}{\partial x_j}\right|\left|\frac{\partial\phi}{\partial x_j}\right|\,\mathrm{d}t
\\[.3em] &&
\nonumber
\le E\int_\omega|\varphi|^2\,\mathrm{d}t-\int_\omega|\nabla'f|^2|\phi|^2 \,\mathrm{d}t+2 \sum_{j=1}^2\sqrt{\int_\omega f^2\left|\frac{\partial\phi}{\partial x_j}\right|^2\,\mathrm{d}t\int_\omega|\phi|^2\left(\frac{\partial f}{\partial x_j}\right)^2\,\mathrm{d}t}
\\[.3em] &&
\nonumber
\le E\int_\omega|\varphi|^2\,\mathrm{d}t-\int_\omega|\nabla'f|^2|\phi|^2\,\mathrm{d}t+ \sum_{j=1}^2\int_\omega f^2\left|\frac{\partial\phi}{\partial x_j}\right|^2\,\mathrm{d}t+\sum_{j=1}^2\int_\omega|\phi|^2\left|\frac{\partial f}{\partial x_j}\right|^2\,\mathrm{d}t
\\[.3em] &&
\label{eq2.2}
=E\int_\omega|\varphi|^2\,\mathrm{d}t+\int_\omega f^2|\nabla'\phi|^2\,\mathrm{d}t\,.
\end{eqnarray}
Armed with these estimates, we are going to find the indicated lower bound to $H(s)$. To this aim we first rewrite the corresponding quadratic form as
\begin{eqnarray}
\lefteqn{\nonumber
\left\langle H(s)\varphi,\varphi\right\rangle=\int_\omega|\nabla'\phi|^2f^2\,\mathrm{d}t
+\beta_0^2\int_\omega f^2|\partial_\alpha\phi|^2\,\mathrm{d}t-\frac{1}{\alpha_{\mu,\beta_0}^2}\int_\omega
\dot{\mu}\frac{\partial_\alpha f}{f}|\varphi|^2\,\mathrm{d}t}
\\[.3em] &&
\label{0}
\qquad +\frac{1}{\alpha_{\mu,\beta_0}^2}\int_\omega\mu(2\beta_0-\mu)\frac{\partial_\alpha^2f}
{f}|\varphi|^2\,\mathrm{d}t \phantom{AAAAAAAAAAAAAAAAA}
\end{eqnarray}
for any function $\varphi=f\phi\in \mathcal{H}^2(\omega)\cap\mathcal{H}_0^1(\omega)$. In order to establish this relation we have to check that
\begin{equation}
\label{1}
\int_\omega\left(-\Delta_D^\omega\varphi-\beta_0^2\partial_\alpha^2\varphi
-E\varphi\right)\overline{\varphi}\,\mathrm{d}t=\int_\omega|\nabla'\phi|^2f^2\,\mathrm{d}t
+\beta_0^2\int_\omega f^2|\partial_\alpha\phi|^2\,\mathrm{d}t.
\end{equation}
First we integrate by parts
\begin{eqnarray*}
\lefteqn{\int_\omega f(x,y)\overline{\phi(x,y)}\frac{\partial f}{\partial x}(x,y)\frac{\partial\phi}{\partial x}(x,y)\,\mathrm{d}x\,\mathrm{d}y
=-\int_\omega\left(\frac{\partial f}{\partial x}\right)^2(x,y)|\phi(x,y)|^2\,\mathrm{d}x\,\mathrm{d}y}
\\[.3em] &&
\hspace{-1.2em}-\int_\omega f(x,y)\phi(x,y)\frac{\partial\overline{\phi}}{\partial x}(x,y)
\frac{\partial f}{\partial x}(x,y)\,\mathrm{d}x\,\mathrm{d}y
-\int_\omega f(x,y)|\phi(x,y)|^2\frac{\partial^2 f}{\partial x^2}(x,y)\,\mathrm{d}x\,\mathrm{d}y\,,
\end{eqnarray*}
and the same for the variable $y$ which yields
\begin{equation}
\label{3}
\int_\omega f\overline{\phi}\,\frac{\partial f}{\partial x_j}\,\frac{\partial\phi}{\partial x_j}\,\mathrm{d}t +
\int_\omega f\phi\,\frac{\partial f}{\partial x_j}\,\frac{\partial\overline{\phi}}{\partial x_j}\,\mathrm{d}t
=-\int_\omega\left(\frac{\partial f}{\partial x_j}\right)^2|\phi|^2\,\mathrm{d}t -\int_\omega f
|\phi\,|^2\frac{\partial^2f}{\partial x_j^2}\,\mathrm{d}t
\end{equation}
for $j=1,2$. Using polar coordinates $t=(r,\alpha)$ one can check is a similar way that
\begin{eqnarray}
\lefteqn{\nonumber
\int_\omega f\overline{\phi}\,(\partial_\alpha\phi)\,(\partial_\alpha f)\,\mathrm{d}t
+\int_\omega f\phi\,(\partial_\alpha\overline{\phi})\,(\partial_\alpha f)(x,y)\,\mathrm{d}t}
\\[.3em] &&
\label{4}
=-\int_\omega(\partial_\alpha f)^2|\phi|^2\,\mathrm{d}t -\int_\omega f|\phi|^2(\partial_\alpha^2f)\,\mathrm{d}t\,.
\end{eqnarray}
These identities make it possible to infer that
\begin{eqnarray}
\lefteqn{ \nonumber
\int_\omega\left(-\Delta_D^\omega\varphi-\beta_0^2\partial_\alpha^2\varphi
-E\varphi\right)\overline{\varphi}\,\mathrm{d}t}
\\[.3em] &&
\nonumber
= \sum_{j=1}^2 \int_\omega\left|\frac{\partial f}{\partial x_j}\phi+\frac{\partial\phi}{\partial x_j}f\right|^2\,\mathrm{d}t
+\beta_0^2\int_\omega\left|f(\partial_\alpha\phi)+
\phi(\partial_\alpha f)\right|^2\,\mathrm{d}t-E\int_\omega f^2|\phi|^2\,\mathrm{d}t
\\[.3em] &&
\nonumber
= \sum_{j=1}^2 \bigg( \int_\omega\left(\frac{\partial f}{\partial x_j}\right)^2|\phi|^2\,\mathrm{d}t
+\int_\omega\bigg|\frac{\partial\phi}{\partial x_j}\bigg|^2f^2\,\mathrm{d}t
+\int_\omega f\overline{\phi}\frac{\partial f}{\partial x_j}\frac{\partial\phi}{\partial x_j}\,\mathrm{d}t
\\[.3em] &&
\nonumber
\quad +\int_\omega f\phi\frac{\partial f}{\partial x_j}\frac{\partial\overline{\phi}}{\partial x_j}\,\mathrm{d}t \bigg)
+\beta_0^2\bigg(\int_\omega f^2|\partial_\alpha\phi|^2\,\mathrm{d}t+\int_\omega|\phi|^2(\partial_\alpha f)^2\,\mathrm{d}t\\[.3em] &&
\nonumber
\quad +\int_\omega f\overline{\phi}(\partial_\alpha\phi)(\partial_\alpha f)\,\mathrm{d}t
+\int_\omega f\phi(\partial_\alpha\overline{\phi})(\partial_\alpha f)\,\mathrm{d}t\bigg)
-E\int_\omega f^2|\phi|^2\,\mathrm{d}t
\end{eqnarray}
Using then the identities (\ref{3}) and (\ref{4}) we rewrite the last expression as
\begin{eqnarray*}
\lefteqn{\sum_{j=1}^2 \left(\int_\omega f^2\left|\frac{\partial\phi}{\partial x_j}\right|^2\,\mathrm{d}t+\int_\omega|\phi|^2\left(\frac{\partial f}{\partial x_j}\right)^2
\,\mathrm{d}t-\int_\omega|\phi|^2\left(\frac{\partial f}{\partial x_j}\right)^2\,\mathrm{d}t-\int_\omega f|\phi|^2\frac{\partial^2f}
{\partial x_j^2}\,\mathrm{d}t \right)}
\\[.3em] &&
+\beta_0^2 \left(\int_\omega f^2|\partial_\alpha\phi|^2\,\mathrm{d}t+\int_\omega|\phi|^2(\partial_\alpha f)^2\,\mathrm{d}t
-\int_\omega|\phi|^2(\partial_\alpha f)^2\,\mathrm{d}t-\int_\omega f|\phi|^2(\partial_\alpha^2f)\,\mathrm{d}t \right)
\\[.3em] &&
\quad -E\int_\omega f^2|\phi|^2\,\mathrm{d}t
\\[.3em] &&
=\int_\omega\Big(f^2|\nabla'\phi|^2-f|\phi|^2\Delta_D^\omega f+\beta_0^2f^2|\partial_\alpha\phi|^2-\beta_0^2f|\phi|^2(\partial_\alpha^2f)-Ef^2|\varphi|^2\Big)\,\mathrm{d}t\,,
\end{eqnarray*}
and since $-\Delta_D^\omega f-\beta_0^2\partial_\alpha^2f-Ef=0$ holds by assumption, we arrive at the relation (\ref{1}). Using now the estimates (\ref{perturb1}) and (\ref{eq2.2}) we infer from (\ref{0}) that
\begin{eqnarray*}
\lefteqn{\left\langle H(s)\varphi,\varphi\right\rangle\ge\int_\omega|\nabla'\phi|^2f^2\,\mathrm{d}t+\beta_0^2\int_\omega f^2|\partial_\alpha\phi|
^2\,\mathrm{d}t-\frac{|\dot{\mu}(s)|}{2\alpha_{\mu,\beta_0}^2}\left(\int_\omega f^2|\partial_\alpha\phi|^2\,\mathrm{d}t
+\int_\omega|\varphi|^2\,\mathrm{d}t\right)}
\\[.3em] &&
\quad -\frac{\mu(s)(2\beta_0-\mu(s))}{\alpha_{\mu,\beta_0}^2\beta_0^2}\left(E\int_\omega|\varphi|^2\,\mathrm{d}t+\int_\omega f^2|\nabla'\phi|^2\,
\mathrm{d}t\right)
\\[.3em] &&
\ge\left(1-\frac{\mu(s)(2\beta_0-\mu(s))}{\alpha_{\mu,\beta_0}^2\beta_0^2}\right)\int_\omega f^2|\nabla'\phi|^2\,\mathrm{d}t
+\left(\beta_0^2-\beta_0^2\frac{|\dot{\mu}(s)|}{2\alpha_{\mu,\beta_0}^2\beta_0^2}\right) \int_\omega f^2|\partial_\alpha\phi|^2\,\mathrm{d}t
\\[.3em] &&
\quad -\frac{E\mu(s)(2\beta_0-\mu(s))}{\alpha_{\mu,\beta_0}^2\beta_0^2}\int_\omega|\varphi|^2\,\mathrm{d}t-\frac{|\dot{\mu}(s)|}
{2\alpha_{\mu,\beta_0}^2} \int_\omega |\varphi|^2\,\mathrm{d}t
\\[.3em] &&
\ge\left(1-\frac{2\mu(s)}{\alpha_{\mu,\beta_0}^2\beta_0}\right)\int_\omega f^2|\nabla'\phi|^2\,\mathrm{d}t+ \beta_0^2\left(1-
\frac{|\dot{\mu}(s)|}{2\alpha_{\mu,\beta_0}^2\beta_0^2}\right) \int_\omega f^2|\partial_\alpha\phi|^2\,\mathrm{d}t
\\[.3em] &&
\quad -\frac{2E\mu(s)}{\alpha_{\mu,\beta_0}^2\beta_0}\int_\omega|\varphi|^2\,\mathrm{d}t
-\frac{|\dot{\mu}(s)|}{2\alpha_{\mu,\beta_0}^2}
\int_\omega|\varphi|^2\,\mathrm{d}t
\\[.3em] &&
\ge\left(1-\max\left\{\frac{2\mu(s)}{\alpha_{\mu,\beta_0}^2\beta_0},\,\frac{|\dot{\mu}(s)|}{2\alpha_{\mu,\beta_0}^2\beta_0^2}\right\}\right)
\int_\omega\left(f^2|\nabla'\phi|^2\,\mathrm{d}t+\beta_0^2\int_\omega f^2|\partial_\alpha\phi|^2\,\mathrm{d}t\right)
\\[.3em] &&
\quad -\frac{2E\mu(s)}{\alpha_{\mu,\beta_0}^2\beta_0}\int_\omega|\varphi|^2\,\mathrm{d}t-\frac{|\dot{\mu}(s)|}{2\alpha_{\mu,\beta_0}^2}
\int_\omega|\varphi|^2\,\mathrm{d}t\,.
\end{eqnarray*}
Using now (\ref{1}) we arrive at the operator inequality
$$
H(s)\ge\left(1-\max\left\{\frac{2\mu(s)}{\alpha_{\mu,\beta_0}^2\beta_0},\,\frac{|\dot{\mu}(s)|}
{2\alpha_{\mu,\beta_0}^2\beta_0^2}\right\}\right)h_{\beta_0}-\frac{2E\mu(s)}{\alpha_{\mu,\beta_0}^2\beta_0} -\frac{|\dot{\mu}(s)|}{2\alpha_{\mu,\beta_0}^2}
$$
from which the sought discreteness of $\sigma(H(s))$ readily follows.
\end{proof}

\section{Dependence on the cross section} \label{s: example}
\setcounter{equation}{0}

We conclude the paper by a couple of examples illustrating how the obtained bound depends on the cross section shape.

\subsection{First example: an elliptic disc} \label{ss: 1st}

Given a positive number $\varepsilon$ we assume that $\omega_\varepsilon$ is an elliptic disc the boundary of which is described by the relation  $(1+\varepsilon)^2x^2+y^2=1$. We are going to show that for sufficiently small values of the parameter $\varepsilon$ and a gentle perturbation $\mu(\cdot)$ the negative-spectrum moments of the operator $H_\beta-E$
on the corresponding twisted tube (\ref{Tube}) behave asymptotically as
$$
\mathcal{O}\left(\varepsilon^{\sigma+1/2}\right)\int_{-s_0}^{s_0}\big(|\dot{\mu}(s)|+\mu(s)
(2\beta_0-\mu(s))\big)^{\sigma+1/2}\,\mathrm{d}s\,.
$$
To prove this claim, let us estimate the right-hand side of (\ref{Inequality}) in this case. First we check that the negative spectrum of operator $H(s)$ is nonempty under the condition $\|\mu\|_\infty<2\beta_0$. Using the fact that $-\Delta_D^\omega f-\beta_0^2 \partial_\alpha^2f-Ef=0$ holds by assumption, that the function $f$ is strictly positive on $\omega_\varepsilon$ according to \cite{EK05}, and the relation $\int_{\omega_\varepsilon}(\partial_\alpha f)f\,\mathrm{d}t=0$, we find
\begin{eqnarray*}
\lefteqn{(f,H(s)f)_{L^2(\omega_\varepsilon)}=\frac{\mu(s)(2\beta_0-\mu(s))}{\alpha_{\mu,\beta_0}^2}\int_{\omega_\varepsilon}
(\partial_\alpha^2f) f\,\mathrm{d}t} \\[.3em] && =-\frac{\mu(s)(2\beta_0-\mu(s))}{\alpha_{\mu,\beta_0}^2}\int_{\omega_\varepsilon}(\partial_\alpha f)^2 \,\mathrm{d}t<0\,, \phantom{AAAAAAAA}
\end{eqnarray*}
hence the claim follows by minimax principle. On the other hand, the lower bound to $H(s)$ proved in the previous section under the conditions (\ref{sp}) means that the negative spectrum consists of a finite number of negative eigenvalues, the multiplicity taken into account. Moreover, their number has an upper bound independent of $\varepsilon$, and in fact, $H(s)$ has for small $\varepsilon$ a single negative eigenvalue as we are going to check next.

We use \emph{reduction ad absurdum} assuming that there is more than one negative eigenvalue. Let $\{\lambda_k(s)\}_{k=1}^M$ with $M>1$ be the negative spectrum of $H(s)$ and $\{g_k(s)\}_{k=1}^M$ the family of corresponding eigenfunctions. It is easy to see that
\begin{eqnarray}
\label{lambda}
\lefteqn{\lambda_k(s)=\int_{\omega_\varepsilon}((h_{\beta_0}-E)g_k(s))\overline{g}_k(s)\,\mathrm{d}t}
\\[.3em] &&
-\frac{1}{\alpha_{\mu,\beta_0}^2}
\int_{\omega_\varepsilon}\left(\dot{\mu}(s)\frac{(\partial_\alpha f)|g_k(s)|^2}{f}-\mu(s)(2\beta_0-\mu(s))\frac{(\partial_\alpha^2f)|g_k(s)|^2}{f}\right)\,\mathrm{d}t\,; \nonumber
\end{eqnarray}
we shall estimate the second term on the right-hand side. We write the operator $h_{\beta_0}$ more explicitly as
$$
h_{\beta_0} =-\frac{\partial^2}{\partial x^2}-
\frac{\partial^2}{\partial y^2} -\beta_0^2\left(x^2\frac{\partial^2}{\partial y^2}
-x\frac{\partial}{\partial x}-2xy\frac{\partial^2}{\partial x\partial y}-
y\frac{\partial}{\partial y}+y^2\frac{\partial^2}{\partial x^2}\right)\,.
$$
and pass to the coordinates $x=\frac{\xi}{1+\varepsilon},\,y=\eta$ which allows us to replace the operator $h_{\beta_0}$ by a unitarily equivalent operator on $L^2\left(\omega_0,\frac{1}{1+\varepsilon}\,\mathrm{d} \xi\,\mathrm{d}\eta\right)$ acting as
\begin{equation}
\label{ellipse}
\widetilde{h}_{\beta_0}=-\Delta_D^{\omega_0}-\beta_0^2\partial_\alpha^2+\varepsilon\left(-(2+\varepsilon)(\beta_0^2\eta^2+1)
\frac{\partial^2}{\partial\xi^2}+\frac{2+\varepsilon}{(1+\varepsilon)^2}\beta_0^2\xi^2\frac{\partial^2}{\partial\eta^2}\right)\,.
\end{equation}
Using the fact that the principal eigenfunction $f$ of $h_{\beta_0}$ on $L^2(\omega_\varepsilon)$ equals
$$
f(x,y)=\widetilde{f}\left((1+\varepsilon)x,y\right),\quad(x,y)\in\omega_\varepsilon\,,
$$
where  $\widetilde{f}$ is the ground-state eigenfunction of $\widetilde{h}_{\beta_0}$ on $L^2\left(\omega_0, \frac{1}{1+\varepsilon}\, \mathrm{d}\xi\,\mathrm{d}\eta\right)$ and applying the first-order perturbation theory \cite{K66} one gets that for small values of the parameter $\varepsilon$
\begin{equation}
\label{perturb_derivative}
f(x,y)=f_\mathrm{disc}\left((1+\varepsilon)x,y\right)+
\mathcal{O}(\varepsilon)\,,
\end{equation}
where $f_\mathrm{disc}$ is the ground-state eigenfunction of operator $-\Delta_D^{\omega_0}-\beta_0^2\partial_\alpha
^2$ on $L^2(\omega_0)$, which coincides naturally with the ground-state eigenfunction of the Dirichlet Laplacian on the unit circle.

Let us next show show that there is a constant $C>0$ such that the bound
\begin{equation}
\label{estimate-perturbation}
\frac{|g_k(s)|}{f}\le C\,, \quad k=1,\ldots,M\,,
\end{equation}
holds on the cross-section $\omega_\varepsilon$ uniformly in $s$. Passing to polar coordinates and denoting
$g^*(r,\varphi):=g(r\cos\varphi,r\sin\varphi)$ for any $k=1,\ldots,M$ and similarly for $f$, and using the Dirichlet condition at the boundary of $\omega_\varepsilon$, one gets
\begin{eqnarray*}
\lefteqn{\frac{|g_k(s)(x,y)|}{f(x,y)}=\frac{|g^*_k(s)(r,\varphi)|}{f^*(r,\varphi)}=\frac{|g^*_k(s)(r,\varphi)-g^*_k(s)(r(\varphi),
\varphi)|}{f^*(r,\varphi)-f^*(r(\varphi),\varphi)}} \\[.3em] &&
\hspace{3.6em} =-\frac{|(\partial g^*_k(s)/\partial r)(r_0(\varphi),\varphi)|}{(\partial f^*/\partial r)(r_1(\varphi),\varphi)}\,, \phantom{AAAAAAAAAA}
\end{eqnarray*}
where the curve $\varphi\mapsto r(\varphi)$ with $\varphi\in[0,2\pi)$ is the boundary of the ellipse $\omega_\varepsilon$ and $r_0(\varphi),r_1(\varphi)\in[0,r(\varphi)]$. Note that the left-hand side of (\ref{estimate-perturbation}) is well defined and continuous on the precompact set $\mathrm{supp}\,\mu \times \omega_\varepsilon$, hence we have to care only about the behavior in the vicinity of the boundary. We have $\left|\frac{\partial f_\mathrm{disc}}{\partial r}\right|>\alpha>0$ near the boundary of the unit circle and using the asymptotics (\ref{perturb_derivative}) for points near the boundary of the ellipse $\omega_\varepsilon$ we get
$$
\frac{|g_k(s)(x,y)|}{f(x,y)}\le\frac{1}{\alpha+\mathcal{O}(\varepsilon)}\left\|\frac{\partial
g^*_k(s)}{\partial r}\right\|_{L^\infty(\omega_\varepsilon)}\,,
$$
which yields the bound (\ref{estimate-perturbation}) valid uniformly is $s$.

Working out the perturbation theory in (\ref{perturb_derivative}) to higher orders and using the fact that the angular derivative of the function $f_\mathrm{disc}$ is zero one can check that
\begin{equation}
\label{&&&}
\partial_\alpha f=\mathcal{O}(\varepsilon)\,,\quad\partial_\alpha^2f=\mathcal{O}(\varepsilon)\,,
\end{equation}
and combining this asymptotics with the inequality (\ref{estimate-perturbation}) we find that the second term on the right-hand side of (\ref{lambda}) behaves as $\mathcal{O}(\varepsilon)$ in the limit $\varepsilon\to 0$. Since the first term is positive and $\lambda_k(s)<0$ holds for $k=1,\ldots,M$ by assumption, it follows
\begin{equation}
\label{h}
\int_{\omega_\varepsilon}\left(\left(h_{\beta_0}-E\right)g_k(s)\right)\overline{g_k}(s)\,\mathrm{d}t=\mathcal{O}(\varepsilon)\,.
\end{equation}
Now we expand the functions $g_k(s),\:k=1,\ldots,M$, in the orthonormal basis $\{f^m\}_{m=1}^\infty$ corresponding to the eigenvalues $\{E_m\}_{m=1}^\infty$ of the operator $h_{\beta_0}-E$ on $L^2(\omega_\varepsilon)$, arranged in the ascending order, i.e.
$$
g_k(s)=\sum_{m=1}^\infty c_{k,m}(s)f^{m}\,.
$$
The principal eigenvalue $E$ of the operator $h_{\beta_0}$ is simple. Indeed, assume that there are two eigenfunctions $q_1$ and $q_2$ such that $h_{\beta_0}q_j=Eq_j$ holds for $j=1,2$. By the first order of perturbation theory and (\ref{ellipse}) we then have
$$
\label{asymptotic behavior}
q_j(x,y)=f_\mathrm{disc}((1+\varepsilon)x,y)+\mathcal{O}(\varepsilon)\,, \quad j=1,2\,,
$$
which, however, contradicts to the orthonormality of the functions $q_1$ and $q_2$. Inserting now the Fourier expansion into  (\ref{h}) we get
$$
\int_{\omega_\varepsilon}\left(\sum_{m=2}^\infty c_{k,m}(s)E_mf^m\right)\overline{\left(\sum_{l=1}^\infty c_{k,l}(s)f^l\right)}\,\mathrm{d}t
=\mathcal{O}(\varepsilon)\,,
$$
and consequently, $\sum_{m=2}^\infty E_m|c_{k,m}(s)|^2=\mathcal{O}(\varepsilon)$, and since $E_2>0$, we have
$$
\sum_{m=2}^\infty|c_{k,m}(s)|^2=\mathcal{O}(\varepsilon)\,,
$$
which means that
$$
\int_{\omega_\varepsilon}\left|g_k(s)-c_{k,1}(s)f\right|^2\,\mathrm{d}t=\mathcal{O}(\varepsilon)
$$
as $\varepsilon\to 0$. This in turn implies
$$
g_k(s)=c_{k,1}(s)f+\mathcal{O}(\varepsilon)\,,\quad k=1,\ldots,M\,,
$$
a.e. in $\omega_\varepsilon$. The orthogonality of the functions $g_k(s)$ then requires
$$
c_{k,1}(s)c_{l,1}(s)=\mathcal{O}(\varepsilon)\quad\text{if}\;\; k\ne l\,,
$$
which contradicts to the fact that $c_{k,1}(s)=1+\mathcal{O}(\varepsilon)$ coming from the relation  $\int_{\omega_\varepsilon}|g_k(s)|^2\,\mathrm{d}t=1$. Consequently, the operator $H(s)$ has a single negative eigenvalue for all $\varepsilon>0$ small enough.

Now it is sufficient to combine relations (\ref{lambda}), (\ref{estimate-perturbation}), and (\ref{&&&}) with the fact that $h_{\beta_0}-E\ge0$ to arrive at the asymptotics
$$
\lambda_1(s)=\mathcal{O}(\varepsilon)\left(|\dot{\mu}(s)|+\mu(s)(2\beta_0-\mu(s))\right)
$$
for $\varepsilon\to 0$ which proves the sought assertion.

\subsection{Second example: a multiply folded ribbon}

Next we are going to illustrate the opposite effect, namely that choosing the cross section far from the circular shape we can make for a fixed twist perturbation $\mu$ the right-hand side of the estimate (\ref{Inequality}) arbitrarily large. To this aim we shall construct an appropriate sequence of thesets $\omega_k,\:k=1,2,\ldots$.

We begin with the circles $T_1$ and $T_2$ centered at zero of radii one and two, respectively. We cross them by axes and quadrant axes, i.e. the lines $y=\pm x$, and connect the intersection points by a closed piecewise linear zigzag curve denoted as $\Gamma_\mathrm{out}^{(1)}$. Next we fix a small positive number $\varepsilon$ and construct in a similar way the curve $\Gamma_\mathrm{in}^{(1)}$ with vertices on the circles of radii $1-\varepsilon$ and $2-\varepsilon$. The two curves are the inner and outer boundary of the set which we denote as $\omega_1$ -- cf. Fig.~1.
\begin{figure}[htb]
\begin{center}
\includegraphics[width=1in,keepaspectratio,angle=0,trim= 0 50 0 0]{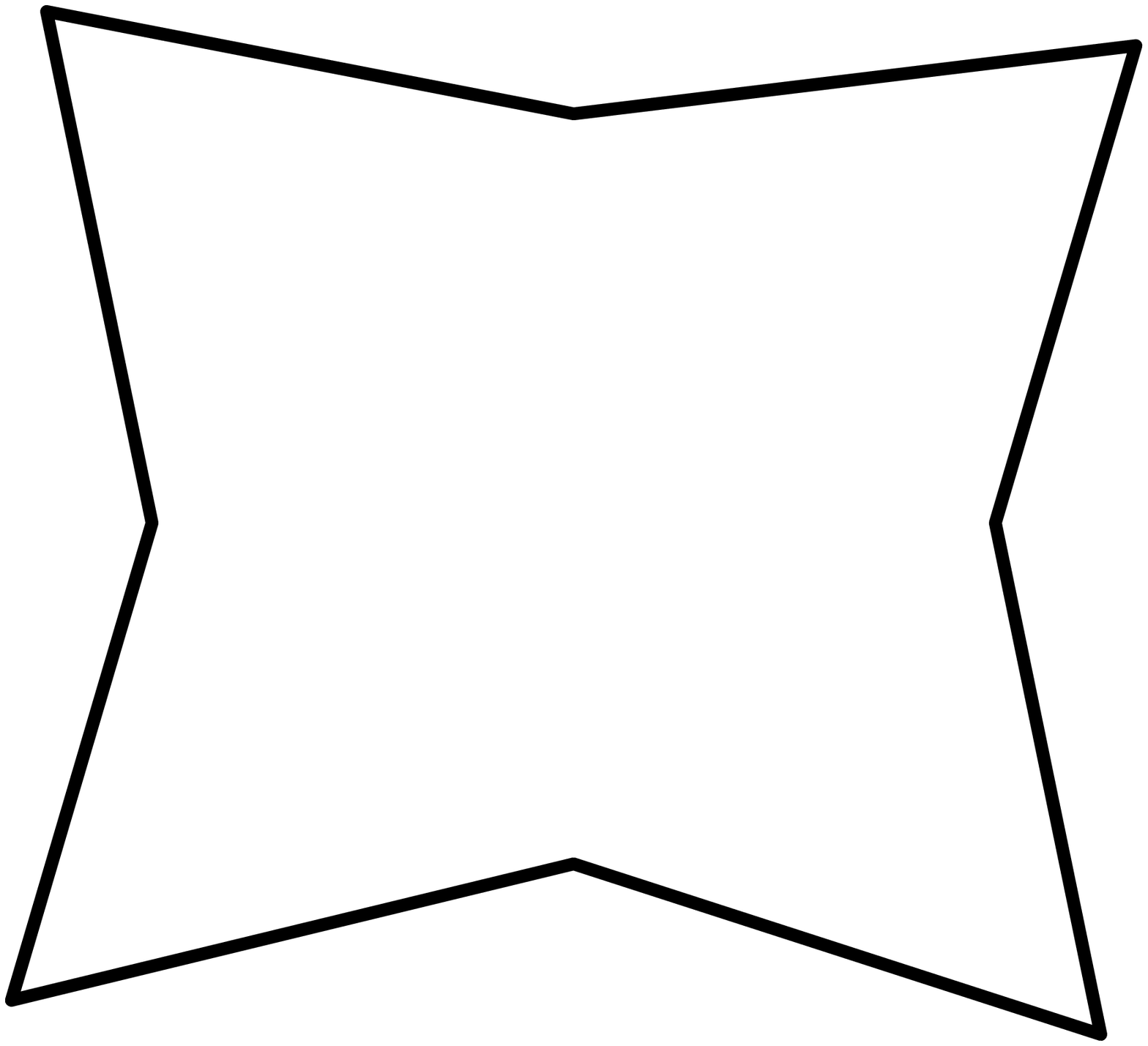}
\hspace{5em}
\includegraphics[width=1in,angle=0,]{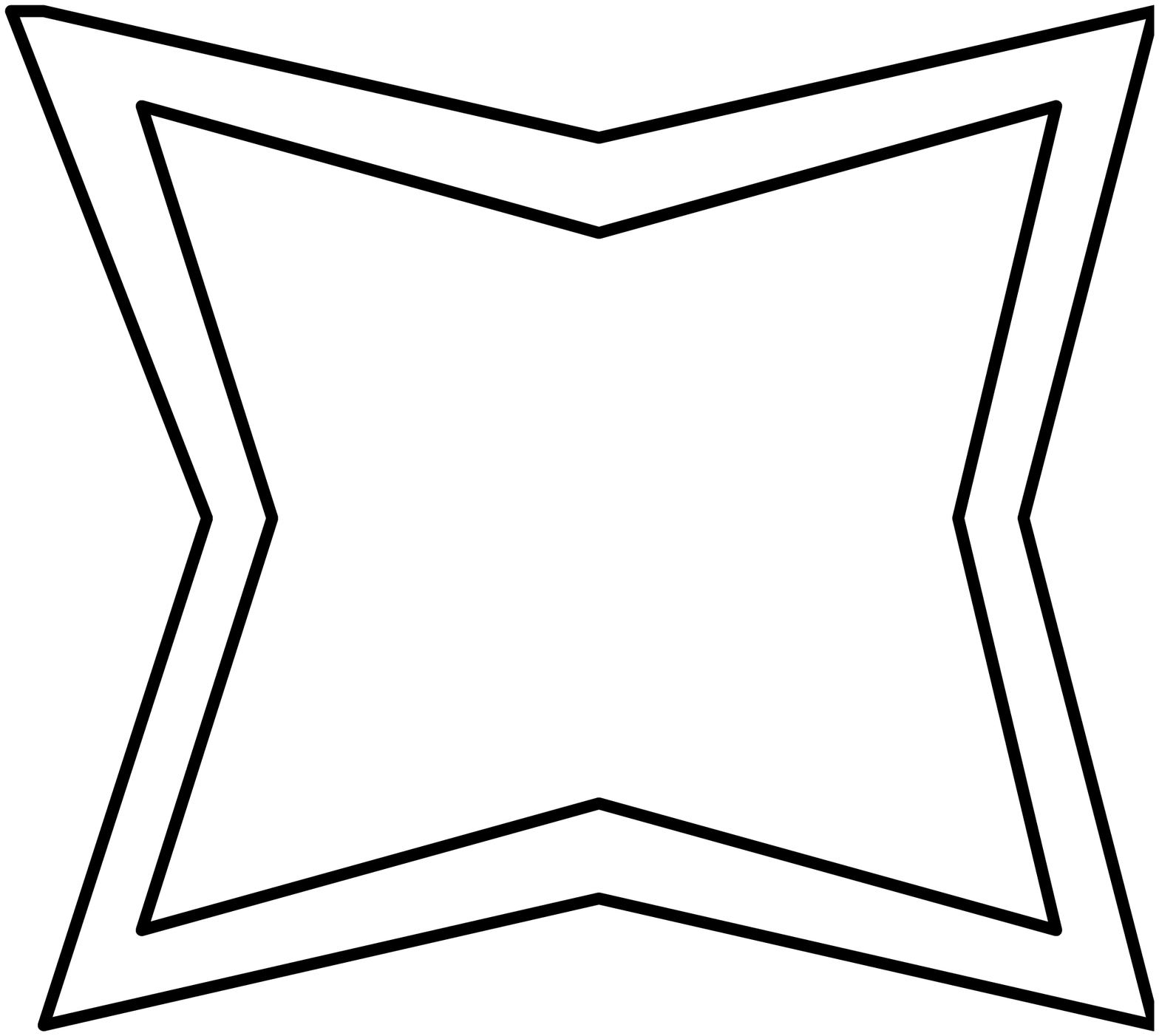}
\caption{The curve $\Gamma_\mathrm{in}^{(1)}$ and the set $\omega_1$.}
\end{center}
\end{figure}
Next we construct the set $\omega_2$ bounded by the curves $\Gamma_\mathrm{out}^{(2)}$ and $\Gamma_\mathrm{in}^{(2)}$ obtained in a similar way by cutting the plane into sixteen radial segments -- cf. Fig.~2.
\begin{figure}[htb]
\begin{center}
\includegraphics[height=1.5in,width=1in,angle=0]{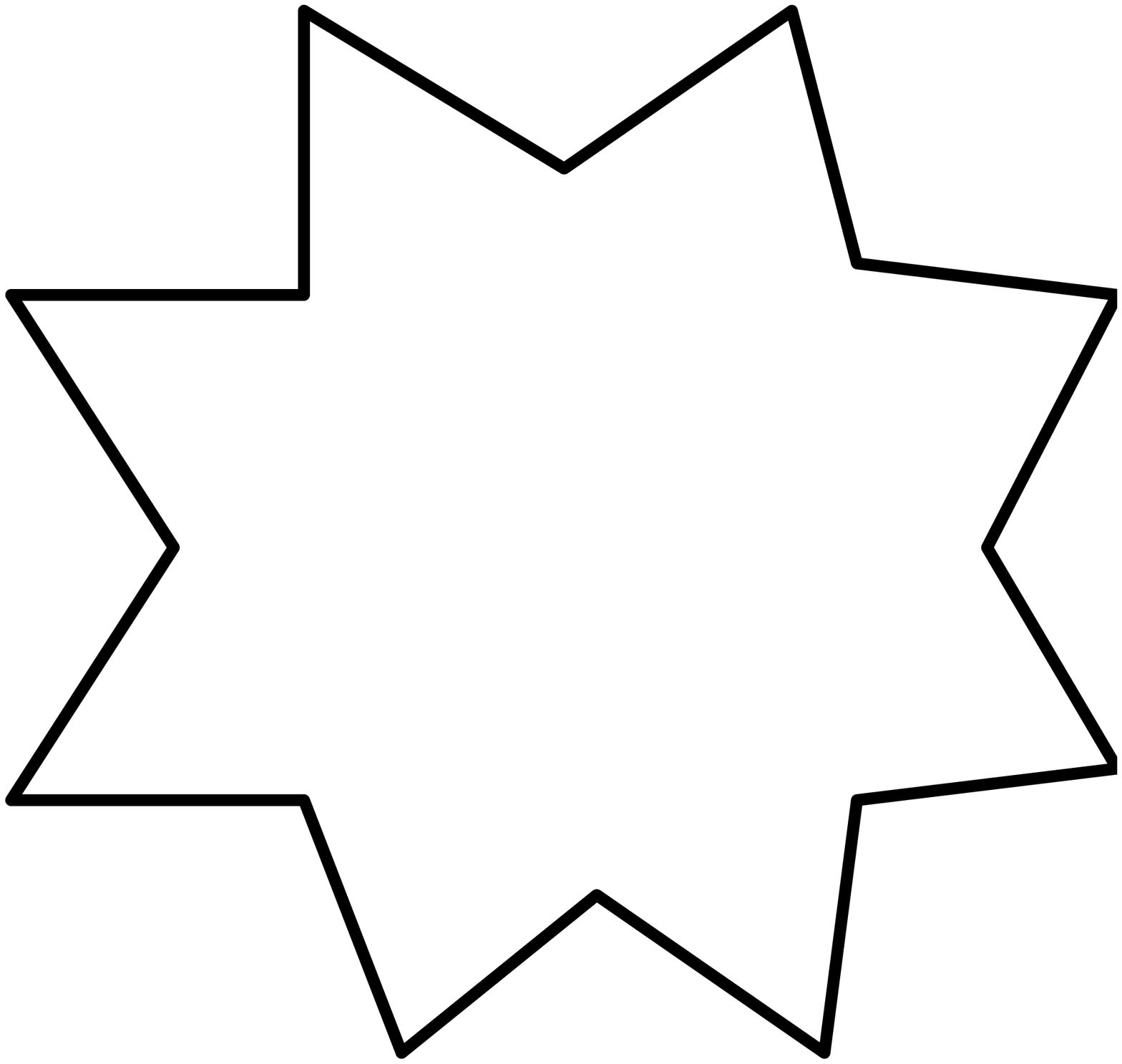}
\caption{The curve $\Gamma_\mathrm{out}^{(2)}$.}
\end{center}
\end{figure}
The construction proceeds in the same way: in the $k$-th step we obtain a `zigzag ribbon' loop with $2^{k+1}$ outer vertices.

Our aim is now to show that if such a `maccaroni' tube is twisted with a fixed slowdown perturbation $\mu$, the right-hand side of the spectral estimate (\ref{Inequality}) will be bound from below by
\begin{equation}
\label{thorny  bound}
\frac{1}{\alpha_{\mu,\beta_0}}\left(\frac{4^{k+1}}{\pi^2}\right)^{\sigma+1/2}\int_{\mathbb{R}}
\mu(s)^{\sigma+1/2}(2\beta_0-\mu(s))^{\sigma+1/2}\,\mathrm{d}s.
\end{equation}
To demonstrate this claim we note first that in view of minimax principle the ground state eigenvalue of $H(s)$ equals
$$
\lambda_{1,k}(s)=\inf\left\{ \int_{\omega_k}\overline{g}\,H(s)g\,\mathrm{d}t\,:\; g\in \mathcal{H}_0^1(\omega_k)\cap\mathcal{H}^2(\omega_k),\,\|g\|_{L^2(\omega_k)}=1 \right\},
$$
which means that
$$
\lambda_{1,k}(s)\le-\frac{1}{\alpha_{\mu,\beta_0}^2}\int_{\omega_k}\left(\dot{\mu}\,(\partial_\alpha
f_k)f_k-\mu(2\beta_0-\mu)(\partial_\alpha^2f_k)f_k\right)\,\mathrm{d}t\,,
$$
where $f_k$ is the ground state eigenfunction of $h_{\beta_0}$ on $L^2(\omega_k)$. Using the fact that $\int_{\omega_k}(\partial_\alpha f_k)f_k\,\mathrm{d}t=0$ and an integration by parts similar to the argument used in the previous example we find that
\begin{equation}\label{eigenvalue}
\lambda_{1,k}(s)\le-\frac{\mu(s)(2\beta_0-\mu(s))}{\alpha_{\mu,\beta_0}^2}\int_{\omega_k}(\partial_\alpha f_k)^2\,\mathrm{d}t\,;
\end{equation}
we are going to estimate the last integral from below.

To this aim we consider separately the parts of $\omega_k$ divided by the four circles used in the construction, specifically
$$\omega_k\cap\{1-\varepsilon\le r<1\}\,,\quad \omega_k\cap
\{1<r<2-\varepsilon\}\,,\quad\text{and}\quad\omega_k\cap\{2-\varepsilon<r\le2\}\,,
$$
and moreover, in view of the discrete rotational symmetry we can discuss only one segment of it. Consider therefore first the part $\widetilde{\omega}$ of $\omega_k\cap\{1\le r<2-\varepsilon\}$
bounded by the circles $T_1$ and $T_{2-\varepsilon}$ and the radial lines
$\varphi=0$ and $\varphi=\frac{\pi}{2^{k+1}}$, and estimate the
integral $\int_{\widetilde{\omega}}(\partial_\alpha f_k)^2\,\mathrm{d}t$ from below. The polar coordinates of the points of $\widetilde{\omega}$ are $r\in(1,2-\varepsilon)$ and $\varphi(r)\in(\varphi_1(r), \varphi_2(r))$ with appropriate $\varphi_j(r)\in\left(0,\frac{\pi} {2^{k+1}}\right)$, hence we get
$$
\int_{\widetilde{\omega}}f_k^2\,\mathrm{d}t=\int_1^{2-\varepsilon}\,\mathrm{d}r
\int_{\varphi_1(r)}^{\varphi_2(r)}r\,f_k^2(r,\alpha)\,\mathrm{d}\alpha\,.
$$
Writing $f_k^2(r,\cdot)$ as a primitive function of the radial derivative and using the Dirichlet condition which $f_k$ has to satisfy at $\partial \omega_k$ we rewrite this expression as
$$
\int_{\widetilde{\omega}}f_k^2\,\mathrm{d}t=\int_1^{2-\varepsilon}\,r\,\mathrm{d}r
\int_{\varphi_1(r)}^{\varphi_2(r)}\left(\int_{\varphi_1(r)}^\alpha
\frac{\partial f_k}{\partial\theta}\,\mathrm{d}\theta\right)^2\,\mathrm{d}\alpha\,,
$$
hence using Cauchy inequality we may infer that
\begin{eqnarray*}
\lefteqn{\int_{\widetilde{\omega}}f_k^2\,\mathrm{d}t\le\frac{\pi}{2^{k+1}}
\int_1^{2-\varepsilon}\,\mathrm{d}r\int_{\varphi_1(r)}^{\varphi_2(r)}
\int_{\varphi_1(r)}^\alpha\,r\left(\frac{\partial f_k}{\partial\theta}\right)^2
\,\mathrm{d}\theta\,\mathrm{d}\alpha} \\[.3em] && \le\frac{\pi^2}{4^{k+1}}\int_1^{2-\varepsilon}\, \mathrm{d}r\,\int_{\varphi_1(r)}^{\varphi_2(r)}\,r\left(\frac{\partial f_k}
{\partial\theta}\right)^2\,\mathrm{d}\theta =\frac{\pi^2}{4^{k+1}}
\int_{\widetilde{\omega}}(\partial_\alpha f_k)^2\,\mathrm{d}t\,,
\end{eqnarray*}
which implies
$$
\int_{\widetilde{\omega}}(\partial_\alpha f_k)^2\,
\mathrm{d}t>\frac{4^{k+1}}{\pi^2}\int_{\widetilde{\omega}}f_k^2\,
\mathrm{d}t\,.
$$
In a similar way one estimate from below contributions from the other parts of the set $\omega_k$, thus combining (\ref{eigenvalue}) with the normalization condition $\int_{\omega_k} f_k^2\,\mathrm{d}t=1$ we get
$$
|\lambda_{1,k}(s)|>
\frac{4^{k+1}}{\pi^2}\frac{\mu(s)(2\beta_0-\mu(s))}{\alpha_{\mu,\beta_0}^2},\,
s\in\mathbb{R}\,,
$$
and consequently, the validity of estimate (\ref{thorny  bound}) for the right-hand side of (\ref{Inequality}).


\bigskip

\begin{flushleft}

Pavel Exner and Diana Barseghyan

\smallskip

Doppler Institute for Mathematical Physics and Applied
Mathematics \\ B\v{r}ehov\'{a} 7, 11519 Prague \\ and  Nuclear
Physics Institute ASCR \\ 25068 \v{R}e\v{z} near Prague, Czechia

\smallskip

Email: exner@ujf.cas.cz, dianabar@ujf.cas.cz

\end{flushleft}


\begin{thebibliography}{10}

\bibitem[BKRS09]{BKRS09}
Ph.~Briet, G.~Raikov, H.~Kova\v{r}\'{\i}k, E.~Soccorsi,
Eigenvalue asymptotics in a twisted waveguide,
\emph{Comm. PDE} \textbf{34} (2009), 818--836.

\bibitem[CB96]{CB96}
I.J.~Clark, A.J.~Bracken,
Bound states in tubular quantum waveguides with torsion,
\emph{J. Phys. A: Math. Gen.} \textbf{29} (1996), 4527--4535.

\bibitem[EKK08]{EKK08}
T.~Ekholm, H.~Kova\v{r}\'{\i}k, D.~Krej\v{c}i\v{r}\'{\i}k, A Hardy inequality in twisted waveguides,\emph{Arch. Rational Mech. Anal.} \textbf{188} (2008), 245--264.

\bibitem[EK05]{EK05}
P.~Exner, H.~Kova\v{r}\'{\i}k: Spectrum of the Schr\"{o}dinger operator in a perturbed periodically twisted tube,
\emph{Lett. Math. Phys.} \textbf{73} (2005), 183--192.

\bibitem[ELW04]{ELW04}
P.~Exner, H.~Linde, T.~Weidl: Lieb-Thirring inequalities for geometrically induced bound states, \emph{Lett. Math. Phys.} \textbf{70} (2004), 83--95.

\bibitem[Ka76]{K66}
T.~Kato: \emph{Perturbation Theory for Linear Operators}, 2nd
edition, Springer, Berlin 1976.

\bibitem[LW00]{LW00}
A.~Laptev, T.~Weidl: Sharp Lieb-Thirring inequalities in high dimensions, \emph{Acta Math.} \textbf{184} (2000), 87--100.

\end{thebibliography}
\end{document}